\UseRawInputEncoding
\documentclass[5p]{elsarticle}

\journal{Elsevier}

\usepackage{subfigure}
\usepackage{fancyhdr}
\usepackage{amsmath}
\usepackage{multirow}
\usepackage{amssymb }
\usepackage{color}
\usepackage{graphics} 
\usepackage{graphicx}
\usepackage{amsmath}
\newtheorem{theorem}{\textbf{Theorem}}
\newtheorem{lemma}{\textbf{Lemma}}
\newtheorem{example}{\textbf{Example}}
\newtheorem{corollary}{\textbf{Corollary}}
\newtheorem{remark}{\textbf{Remark}}
\newtheorem{definition}{\textbf{Definition}}

\newtheorem{proposition}{\textbf{Proposition}}

	\newenvironment{proof}{{{\bf Proof:}}}{\hfill $\square$\par}

	\usepackage{multirow} 
	\usepackage{xcolor}
	\usepackage{multirow}
	\usepackage{setspace}
	\usepackage{geometry}
	\usepackage{url}
	\usepackage{subfigure}
	\usepackage{fancyhdr}
	\usepackage{amsmath}
	\usepackage{multirow}
	\usepackage{amssymb }
	\usepackage{color}
	\usepackage{graphics} 
	\usepackage{graphicx}
	\usepackage{algorithm} 
   \usepackage{algorithmic}

	\usepackage{algorithmic} 
	\usepackage{subeqnarray}
	\usepackage{cases}
	\usepackage{threeparttable}

\begin{document}
		
		\begin{frontmatter}
			
			\title{On Real Structured Controllability/Stabilizability/Stability Radius: Complexity and Unified Rank-Relaxation based Methods} 

			
			\author{Yuan Zhang, Yuanqing Xia, and Yufeng Zhan}
			\address{School of Automation, Beijing Institute of Technology, Beijing, China\\~Email: $\emph{\{zhangyuan14,xia\_yuanqing,yu-feng.zhan\}@bit.edu.cn}$} 

		\begin{abstract}

This paper addresses the real structured controllability, stabilizability, and stability radii (RSCR, RSSZR, and RSSR, respectively) of linear systems, which involve determining the distance (in terms of matrix norms) between a (possibly large-scale) system and its nearest uncontrollable, unstabilizable, and unstable systems, respectively, with a prescribed affine structure. This paper makes two main contributions. First, by demonstrating that determining the feasibilities of RSCR and RSSZR is NP-hard when the perturbations have a general affine parameterization, we prove that computing these radii is NP-hard. Additionally, we prove the NP-hardness of a problem related to the RSSR. These hardness results are independent of the matrix norm used. Second, we develop unified rank-relaxation based algorithms for these problems, which can handle both the Frobenius norm and the $2$-norm based problems and share the same framework for the RSCR, RSSZR, and RSSR problems. These algorithms utilize the low-rank structure of the original problems and relax the corresponding rank constraints with a regularized truncated nuclear norm term. Moreover, a modified version of these algorithms can find local optima with performance specifications on the perturbations, under appropriate conditions. Finally, simulations suggest that the proposed methods, despite being in a simple framework, can find local optima as good as several existing methods.

		\end{abstract}
		
		
	\end{frontmatter}

\section{Introduction}
Controllability and stability are two fundamental properties in control system analysis and design \cite{kailath1980linear}. The conventional notions of controllability and stability are binary concepts. In many practical applications, it is desirable to know how robust a system is to preserve controllability and stability against parameter perturbations. This is particularly important in robustness analysis and robust controller design \cite{K.M.UpperSaddleRiverNewJersey1996Robust}. As such, the concepts of stability and controllability radii, i.e., the distance (in terms of matrix norms) between a system and its nearest unstable or uncontrollable systems since their first introduction \cite{hinrichsen1986stability, paige1981properties}, have received considerable attention in the control community \cite{WickDistance1991}.

Formally, for the following (possibly large-scale) linear time invariant system
\begin{equation} \label{StateSpace} \dot x(t) = Ax(t) + Bu(t),\end{equation}
where $x(t)\in {{\mathbb R}^{n}},u(t)\in {{\mathbb R}^{m}}$ are respectively state and input vectors, $A\in {{\mathbb R}^{n \times n}}$, and $B \in {{\mathbb R}^{n \times m}}$,
the controllability and stability radii are respectively defined as
$$\begin{array}{l} r_{con}(A,B) = \min \{ \left\|[A_\Delta,B_\Delta]\right\|: [A_\Delta,B_\Delta] \in {\mathbb F}^{n\times (n+m)}, \\ (A+A_{\Delta},B+B_\Delta){\kern 1pt} {\kern 1pt} {\kern 1pt} {\rm{is}}{\kern 1pt} {\kern 1pt} {\kern 1pt} {\rm{uncontrollable}}\},\end{array}$$
$$r_{stb}(A) = \min \{ \left\|A_\Delta\right\|: A_\Delta \in {\mathbb F}^{n\times { n}}, A+A_{\Delta} {\kern 1pt} {\kern 1pt} {\kern 1pt} {\rm{is}}{\kern 1pt} {\kern 1pt} {\kern 1pt} {\rm{unstable}}\},$$
where ${\mathbb F}={\mathbb C}$ or $\mathbb{R}$, and $||\cdot||$ is the $2$-norm or Frobenius norm (F-norm in abbreviation). A related notion to the controllability radius is the stabilizability radius, which is defined similarly to the former one.

When there is no structural constraint on $[A_\Delta, B_\Delta]$, depending on $\mathbb{F}$, the study of controllability radius could be classified into two categories. For the complex, unstructured controllability radius, \cite{R.E1984Between} revealed an algebraic expression of $r_{con}(A,B)$ (in this case, $r_{con}(A,B)$ in terms of the $2$-norm and F-norm are equivalent). An $O(n^4)$ bisection algorithm for computing $r_{con}(A,B)$ to any prescribed accuracy was given in \cite{gu2006fast}.  For the real, unstructured controllability radius, \cite{Hu2004Real} gave an algebraic formula that enables computing the exact $r_{con}(A,B)$ in terms of the $2$-norm via the grid search over the complex plane. A method based on the structured total least square algorithm on the extended controllability matrix was carried out later in \cite{khare2012computing}.\footnote{The method in \cite{khare2012computing} could be trivially extended to the structured case.} Similarly, for $A_\Delta$ unstructured, the study on stability radius could also be classified into two categories depending on which field the perturbation belongs to. Early researches that investigated the complex, unstructured stability radius include \cite{hinrichsen1986stability,byers1988bisection,boyd1990regularity} and the references in \cite{hinrichsen1990real}. For $A_\Delta$ real but otherwise unstructured, \cite{freitag2014new} proposed an implicit determinant method with quadratic convergence, while \cite{rostami2015new} presented iterative algorithms relying on computing the rightmost eigenvalues of a sequence of matrices.

The simplest case of real structured stability radius (RSSR) problems could date back to \cite{hinrichsen1986stability,qiu1993formula}, where $A_{\Delta}=E\Delta H$, with $E,H$ given and $\Delta$ unknown but {\emph{full}}. For this case, \cite{qiu1993formula} established a well-known $\mu$-value related formula, which enabled computing the exact $r_{stb}(A)$ in terms of the $2$-norm via a global search over the imaginary axis. Later, \cite{qiu1995computation,sreedhar1996fast} extended this approach by limiting the search space to certain intervals or level-sets.

Recently, the more realistic case where $A_\Delta$ depends affinely on the {\emph{perturbed parameters}} has received increasing attention  \cite{hinrichsen1990real,johnson2018structured,bianchin2016observability,katewa2020real}, where $A_{\Delta}=\sum \nolimits_{i=1}^p \theta_iA_i$, for $A_i$ given and $\theta_i\in {\mathbb{R}}$ unknown (
a special case is $A_{\Delta}=E\Delta H$, with $\Delta$ having a prescribed zero-nonzero sparsity pattern, which is called the {\emph{multiple perturbation structure}} in \cite{hinrichsen1990real}). Particularly,  \cite{bianchin2016observability,katewa2020real} presented Lagrange multiplier based characterizations and gradient-based iterative algorithms for the local optima of the real structured controllability radius (RSCR; observability by duality) and RSSR problems. \cite{johnson2018structured} provided an iterative algorithm to find local optima for the smallest additive structured perturbation for which a system property (e.g., controllability, stability) fails to hold. All of these algorithms are tailored for the F-norm based problems.

A remarkable feature of all the publications above is that, they deal with the $2$-norm and F-norm based RSCR and RSSR problems separately, unless in the complex, unstructured case where they are inherently equivalent. This is rather reasonable considering that the $2$-norm and F-norm have considerably different characteristics.  In this paper, instead, we will for the first time propose algorithms that could deal with the two different matrix norm based RSSR (and RSCR) problems in a unified way. Moreover,  for the RSSR and RSCR, similar to \cite{johnson2018structured}, our algorithms share the same framework.

More specifically, following \cite{johnson2018structured,bianchin2016observability,katewa2020real}, we consider the RSSR, RSCR, as well as the real structured stabilizability radius (RSSZR, defined similarly to the RSCR) problems in the case where the perturbations are {\emph{affinely parameterized by the perturbed parameters}}. Such an affine structure is frequently used in describing how the system matrices are affected by the elementary (physical, chemical, or biological) parameters \cite{Morse_1976,Anderson_1982,zhou1996robust,zhang2019structural,zhang2021structural,KarowStructured2009}. Under this framework, we first answer a fundamental problem: what are the computational complexities of those problems. We show that the RSCR and RSSZR with the affine structure are both NP-hard, and even deciding their feasibilities are. This means the hardness results are irrespective of what norms are utilized (even beyond the $2$-norm and F-norm). To our knowledge, no complexity claim has been made for the RSCR and RSSZR before. As for the RSSR, we prove the NP-hardness of checking the feasibility of a problem that is closely related to it; that is, the problem of determining whether an affinely structured matrix has a prescribed eigenvalue in the imaginary axis. This problem is relevant if we are to extend several existing algorithms developed for the RSSR in the literature \cite{qiu1993formula,qiu1995computation,sreedhar1996fast,freitag2014new} to the affine case. The NP-hardness of the addressed RSSR problem with respect to the $2$-norm is obtained from a known result in robust stability analysis \cite{nemirovskii1993several} (but it seems not easy to extend this result to show the NP-hardness of the F-norm based RSSR).

Next, by leveraging the low-rank structure of the RSSR, RSCR, and RSSZR, we propose {\emph {unified}} rank-relaxation based algorithms towards them. Those algorithms are inspired by the truncated nuclear norm based method for the low-rank matrix completions \cite{hu2012fast}, and are guaranteed to converge to stationary points of the relaxed problems. Furthermore, under suitable conditions, we improve our algorithms by adding an initialization stage and specifically designing the regularization parameter, so that they are assured to find local optima with performance specifications on the corresponding perturbations. Our proposed algorithms are valid both for the $2$-norm and F-norm based problems, which is a distinct feature from \cite{johnson2018structured}.  Numerical simulations show that, our algorithms, though in a simple framework, could find local optima as good as several existing methods.

To sum up, the main contributions of this paper are two-fold: we reveal several NP-hardness results arising in the RSSR, RSCR, and RSSZR with the general affine structure, which are irrespective of what matrix norm is adopted; we propose unified rank-relaxation based methods towards the three problems. The rest are organized as follows. Section \ref{problem-formulation} presents the problem formulations. Section \ref{complexity-analysis} provides several NP-hardness results related to the considered problems, while Section \ref{algorithms} gives the rank-relaxation based algorithms respectively for the RSCR, RSSZR, and RSSR. In Section \ref{examples}, several examples are given to demonstrate the effectiveness of the proposed methods. The last section concludes this paper.

\section{Problem formulations} \label{problem-formulation}

\subsection{Notations}
For a matrix $M\in {\mathbb R}^{n_1\times n_2}$, ${\bf \Lambda}(M)$ denotes the set of its eigenvalues (for $M$ square); $\sigma_i(M)$ denotes the $i$th largest singular value of $M$, $i=1,...,\min\{n_1,n_2\}$. We shall also denote the smallest singular value as $\sigma_{\min} (M)$. Let $||M||_F$ and $||M||_2$ be respectively the Frobenius norm (F-norm) and the $2$-norm of matrix $M$.  $I_n$ denotes the $n\times n$ identify matrix, and $0_{n_1\times n_2}$ the $n_1\times n_2$ zero matrix, where the subscripts are omitted if they can be inferred from the context. $|\cdot|$ takes the absolute value of a scalar. $\bf j$ is the imaginary unit.

\subsection{Problem formulations}

Consider a linear system described by (\ref{StateSpace}). It is known that the controllability of system (\ref{StateSpace}) requires ${\rm rank}[A-zI,B]=n$ for all $z\in {\mathbb C}$, while stability requires that all eigenvalues of $A$ have negative real parts \cite{kailath1980linear}. Moreover, the stabilizability of system (\ref{StateSpace}) is the ability to make the system stable using state feedback, requiring all uncontrollable modes of $(A,B)$ to be stable. To measure the robustness of these important system properties against additive perturbations on the system matrices, one common index is the distance (in terms of some matrix norms) between the original system and the set of systems without the corresponding property \cite{R.E1984Between,qiu1993formula}. For practical plants, entries of the state matrix $A$ and input matrix $B$ that can be perturbed may have structural constraints. Particularly, suppose the system matrices are parameterized by the vector of parameters $\theta\doteq [\theta_1,...,\theta_p]^{\intercal}\in {\mathbb R}^p$ as
\begin{equation} \label{affine}
A(\theta)=A+\sum \nolimits_{i=1}^p\theta_i A_i, B(\theta)=B+\sum \nolimits_{i=1}^p\theta_i B_i,
\end{equation}
where $A_i\in{\mathbb{R}}^{n\times n}$, $B_i\in {\mathbb{R}}^{n\times m}$. Define $A_\Delta\doteq \sum \nolimits_{i=1}^p\theta_i A_i$ and $B_\Delta \doteq \sum \nolimits_{i=1}^p\theta_i B_i$ as the corresponding perturbation matrices (perturbations for short). The affine parameterization (\ref{affine}) is common in the literature for describing how system matrices are affected by the elementary parameters $\theta$ \cite{Morse_1976,Anderson_1982,zhou1996robust,zhang2019structural,zhang2021structural,KarowStructured2009}. In particular, the following example shows perturbations of a subset of edges in a network can be naturally described by (\ref{affine}).

\begin{example}[Affine parameterization in networks]Suppose $A(\theta)$ is the Laplacian matrix of a $3$-node directed graph, written as follows
	$$ A(\theta)=\left[\begin{array}{ccc}
	\theta_{13} & 0  & -\theta_{13}  \\
	-\theta_{21} & \theta_{21} & 0\\
	0 & -\theta_{32} & \theta_{32} \\
	\end{array}\right],$$ where $\theta=[\theta_{13},\theta_{21},\theta_{32}]^\intercal$ is the set of weights of edges in the graph. Then, $A(\theta)$ can be rewritten as
	$${{\begin{array}{c} A(\theta)= \left[\begin{array}{ccc}
			1 & 0  & -1  \\
			0 & 0 & 0\\
			0 & 0 & 0 \\
			\end{array}\right]\theta_{13}+\left[\begin{array}{ccc}
			0 & 0  & 0  \\
			-1 & 1 & 0\\
			0 & 0 & 0 \\
			\end{array}\right]\theta_{21}+ \\ \left[\begin{array}{ccc}
			0 & 0  & 0  \\
			0 & 0 & 0\\
			0 & -1 & 1 \\
			\end{array}\right]\theta_{32}.\end{array}}}$$
	Furthermore, consider $A(\theta)$ as the adjacency matrix of a $3$-node undirected graph (without self-loops), written as follows
	$$A(\theta)=\left[\begin{array}{ccc}
	0 & \theta_{12}  & \theta_{13}  \\
	\theta_{12} & 0 & \theta_{23}\\
	\theta_{13} & \theta_{23} & 0 \\
	\end{array}\right].$$Then, $A(\theta)$ can be parameterized  as {\small
		$$A(\theta)\!=\!\left[\begin{array}{ccc}
		0 & 1  & 0  \\
		1 & 0 & 0\\
		0 & 0 & 0 \\
		\end{array}\right]\theta_{12}+\left[\begin{array}{ccc}
		0 & 0  & 1  \\
		0 & 0 & 0\\
		1 & 0 & 0 \\
		\end{array}\right]\theta_{13}+\left[\begin{array}{ccc}
		0 & 0  & 0  \\
		0 & 0 & 1\\
		0 & 1 & 0 \\
		\end{array}\right]\theta_{23}.$$}
\end{example}

Let ${\bf \Gamma}: {\mathbb R}^p\to {\mathbb R}^{t_1\times t_2}$ be a structure specification that maps $\theta$ from the vector space to the $t_1 \times t_2$ dimensional matrix space, such that ${\bf \Gamma}(\theta)\in {\mathbb R}^{t_1\times t_2}$ and each of its entry is an affine combination of $\{\theta_1,...,\theta_p\}\cup {\mathbb R}$. Assume without losing generality that ${\bf\Gamma}(\theta)\ne 0_{t_1\times t_2}$ unless $\theta=0_{p\times 1}$. The real structured controllability, stabilizability, and stability radii are defined respectively as follows:
\\~
\\~
{\bf real structured controllability radius (RSCR):}
$$\begin{array}{l} r_{con}(A(\theta),B(\theta)) = \min \{ \left\|{\bf \Gamma(\theta)}  \right\|: \theta \in {{\mathbb R}^{p}}, (A(\theta),B(\theta)){\kern 1pt} {\kern 1pt}  \\ {\rm{is}}{\kern 1pt} {\kern 1pt} {\kern 1pt} {\rm{uncontrollable}}\} ,\end{array}$$
{\bf real structured stabilizability radius (RSSZR):}
$$\begin{array}{l} r_{stz}(A(\theta),B(\theta)) = \min \{ \left\|{\bf \Gamma(\theta)} \right\|: \theta \in {{\mathbb R}^{p}},  (A(\theta),B(\theta)) \\ {\kern 1pt} {\kern 1pt} {\kern 1pt} {\rm{is}}{\kern 1pt} {\kern 1pt} {\kern 1pt} {\rm{unstabilizable}}\},\end{array}$$
{\bf real structured stability radius (RSSR):}
$$ r_{stb}(A(\theta)) = \min \{ \left\| {\bf \Gamma(\theta)} \right\|: \theta \in {{\mathbb R}^{p}}, A(\theta){\kern 1pt} {\kern 1pt} {\kern 1pt} {\rm{is}}{\kern 1pt} {\kern 1pt} {\kern 1pt} {\rm{unstable }}\} ,$$
in which $||\cdot||$ is the $2$-norm or F-norm. Some typical forms of ${\bf \Gamma}(\theta)$ are
\begin{itemize}
	\item Form (i): ${\bf \Gamma}(\theta)$ is a full matrix whose vectorization (column-wise) is $\theta$;
	\item Form (ii): ${\bf \Gamma}(\theta)$ is a $p\times p$ diagonal matrix whose $i$th diagonal is $\theta_i$;
	\item Form (iii):  ${\bf \Gamma}(\theta)\equiv\theta$.
\end{itemize}It is clear that, for the F-norm, the above three cases are equivalent. In contrast, for the $2$-norm in Form (ii), the corresponding problem then reduces to finding the parameter $\theta$ with its largest element in magnitude as small as possible while the respective property is broken (equivalently, $||{\bf \Gamma}(\theta)||=||\theta||_{\infty}$, with $||\cdot||_{\infty}$ being the $l$-$\infty$ norm). With Form (iii), we have $||{\bf \Gamma}(\theta)||_F=||{\bf \Gamma}(\theta)||_2$.  Hereafter, to ease the description, when the F-norm is used, the corresponding $r_{\star}(\cdot)$ is denoted by $r^F_{\star}(\cdot)$, and $r^S_{\star}(\cdot)$ when the $2$-norm (i.e., the spectral norm) is used, $\star=con, stb$ or $stz$. We may drop `$\cdot$' from $r^S_{\star}(\cdot)$ and $r^F_{\star}(\cdot)$ when the corresponding `$\cdot$' is clear from the context. 

We say the RSCR (RSSZR, RSSR) problem is {\emph{feasible}}, if there exists a $\theta\in {\mathbb{R}}^p$ making the corresponding perturbed system uncontrollable (unstabilizable, unstable).

\section{Complexity analysis} \label{complexity-analysis}
This section provides complexity results of the RSCR and RSSZR  with the general affine structure, as well as a problem that is closely related to the RSSR. Those results are irrespective of what ${\bf \Gamma}(\theta)$ and matrix norm are used.  The complexity of the $2$-norm based RSSR can be obtained from a known result in robust stability analysis.

\begin{theorem} \label{complexity}
	Given $(A(\theta),B(\theta))$ in (\ref{affine}), determining \\~ $r_{con}(A(\theta),B(\theta))$ is NP-hard, and even checking the feasibility of the RSCR problem is NP-complete.
\end{theorem}

To prove Theorem \ref{complexity}, we shall give a reduction from the NP-complete subset sum problem to an instance of the RSCR problem.

\begin{definition}[Subset sum problem \cite{computation_complexity_modern_2009}]
	Given a set $S$ of $p$ integers $S=\{a_1,...,a_p\}$ and a target sum $T\in {\mathbb R}$, the subset sum problem is to determine whether there exists a nonempty subset of $S'\subseteq S$ whose sum is exactly $T$, i.e., $\sum \nolimits_{a_i\in S'} a_i=T$. If the answer is yes, we say this problem is feasible.
\end{definition}

{Our reduction is based on the following observations.
	Given a set $S$ of $p$ integers $S=\{a_1,...,a_p\}$ and a target sum $T\in {\mathbb R}$, an equivalent statement of the subset sum problem is determining whether there is a nonzero vector $\theta\in \{0,1\}^p$ such that $\sum \nolimits_{i=1}^p\theta_i a_i=T$. To restrict $\theta_i\in \{0,1\}$, it suffices to let ${\rm rank}{\tiny \left[\begin{array}{cc} 1-\theta_i & 0 \\ 0 & \theta_i \end{array} \right]}=1$. With these observations, our initial idea in the reduction is to construct the matrix $A(\theta)$ such that its secondary diagonal entries are $T+\sum \nolimits_{i=1}^p\theta_i a_i$, $1-\theta_1$, $\theta_1$,...,$1-\theta_p$, and $\theta_p$, while its main diagonal entries are all zero. Such $A(\theta)$ has only zero eigenvalues. Suppose we have constructed $B(\theta)$, then $(A(\theta),B(\theta))$ is uncontrollable, if and only if ${\rm rank}[A(\theta),B(\theta)]<2p+2$, where $2p+2$ is the number of rows of $A(\theta)$. The construction of $B(\theta)$ should ensure that, the subset sum problem is feasible, if and only if  $[A(\theta),B(\theta)]<2p+2$. To make $B(\theta)$ meet this requirement, it suffices to make every $p+1$ rows of $B(\theta)$ linearly independent. To this end, we resort to the Vandermonde matrix because it exactly possesses such a nice property \cite{R.A.1991Topics}. The details are given as follows. }

{\bf Proof of Theorem \ref{complexity}} Let $S=\{a_1,...,a_p\}$ be a set of $p$ integers and let the target sum $T=-1$. Let $\theta=[\theta_1,...,\theta_p]^{\intercal}$. Construct the $(2p+2)\times (2p+2)$ lower-triangular matrix $A(\theta)$ as
\begin{equation}{\footnotesize\label{A-construct} A(\theta)\!=\!\left[
	\begin{array}{cccccccc}
	0 & 0 & 0 & 0 & \cdots & 0 & 0 & 0 \\
	1+\sum \limits_{i=1}^p \theta_ia_i & 0 & 0 & 0 & \cdots & 0 & 0 & 0 \\
	0 & 1-\theta_1 & 0 & 0 & \cdots & 0 & 0 & 0\\
	0 & 0 & \theta_1 & 0 & \cdots & 0 & 0 & 0\\
	\vdots & \vdots & \vdots & \vdots & \ddots & \vdots & \vdots & \vdots  \\
	0 & 0 & 0 & 0 & \cdots & 1-\theta_P & 0 & 0 \\
	0 & 0  &  0 & 0  & \cdots & 0 & \theta_p & 0 \\
	\end{array}
	\right].}\end{equation}
Let $B(\theta)$ be the $(2p+2)\times (p+1)$ Vandermonde matrix with $2p+2$ distinct scalars $\beta_1,...,\beta_{2p+2}$:
\begin{equation}\label{B-construct}{\small B(\theta)=\left[
	\begin{array}{ccccc}
	1 & \beta_1 & \beta_1^2 & \cdots & \beta_1^p \\
	1 & \beta_2 & \beta_2^2 & \cdots & \beta_2^p \\
	\vdots & \vdots & \vdots & \ddots & \vdots \\
	1 & \beta_{2p+2} & \beta_{2p+2}^2 & \cdots & \beta^{p}_{2p+2} \\
	\end{array}
	\right].}
\end{equation}
We know that for any subset $V\subseteq \{1,...,2p+2\}$ with $|V|=p+1$,  $\det B(\theta)_V= \prod \nolimits_{1\le i< j\le 2p+2, i,j\in V} (\beta_j-\beta_i)$, where $B(\theta)_V$ is the submatrix with rows indexed by $V$. Hence, any $p+1$ rows of $B(\theta)$ are {\emph{linearly independent}}.

We claim that the RSCR problem on $(A(\theta), B(\theta))$ is feasible, if and only if there is a subset $S'\subseteq S$ such that $\sum \nolimits_{a_i\in S'} a_i=-1$.

{\bf If:} Since $A(\theta)$ has only zero eigenvalues, by the PBH test, to make $(A(\theta), B(\theta))$ uncontrollable, it suffices to make ${\rm rank}[A(\theta),B(\theta)]< 2p+2$. If the subset sum problem with $T=-1$ is feasible on a subset $S'\subseteq S$, by setting $\theta_i=1$ if $a_i\in S'$ and $\theta_i=0$ otherwise, we have ${\rm rank} A(\theta)=p$. Then, ${\rm rank}[A(\theta), B(\theta)]\le {\rm rank} A(\theta) + {\rm rank} B(\theta)= p+p+1< 2p+2$, meaning that $(A(\theta), B(\theta))$ is uncontrollable.

{\bf Only if:} If the subset sum problem over $S$ is infeasible, then it is easy to see for all $\theta\in {\mathbb R}^p$, we have ${\rm rank} A(\theta)\ge p+1$.
In this circumstance, by some row permutations on $[A(\theta), B(\theta)]$, it is always possible to get $[A'(\theta), B'(\theta)]$ such that it contains a $(2p+2)\times(2p+2)$ submatrix, which is upper block triangular with two blocks, the first of which consists of $p+1$ rows of $A(\theta)$ with rank $p+1$, and the second of which consists of $p+1$ rows of $B(\theta)$ also with rank $p+1$ (by the Vandermonde matrix construction). Consequently, we get ${\rm rank}[A(\theta),B(\theta)]=2p+2$, for all $\theta\in {\mathbb R}^p$. This leads to the infeasibility of the RSCR problem on $(A(\theta), B(\theta))$.

Since ${\rm rank}A(0)=p+1$, it follows from the {\bf only if} direction that $(A(0),B(0))$ is controllable when $\theta = 0_{p\times 1}$.
Note for a given $\theta\in {\mathbb R}^p$, we can check the controllability of $(A(\theta), B(\theta))$ in polynomial time (for example, checking the full rank of the controllability matrix). Since the reduction above, i.e., the construction of $A(\theta)$ and $B(\theta)$, is in polynomial time, and the subset sum problem is NP-complete, we conclude that checking the feasibility of the RSCR problem is NP-complete. The rest of the statements in Theorem \ref{complexity} follow immediately.  {\hfill $\square$\par}

						An immediate result of the proof for Theorem \ref{complexity} gives the complexity of the RSSZR problem.
						\begin{corollary} \label{complexity-stab}
							Given $(A(\theta),B(\theta))$ in (\ref{affine}), checking the feasibility of the RSSZR problem is NP-hard. Consequently, determining $r_{stz}(A(\theta),B(\theta))$ is NP-hard.
						\end{corollary}
						\begin{proof}
							Note in the proof of Theorem \ref{complexity}, all the eigenvalues of $A(\theta)$ are all zero, which cannot be altered by changing $\theta$. Hence, the controllability of $(A(\theta), B(\theta))$ is equivalent to the stabilizability of $(A(\theta), B(\theta))$ (in fact, one can further change the zero diagonals of $A(\theta)$ in (\ref{A-construct}) to any positive values), indicating that the RSSZR problem and RSCR problem on $(A(\theta), B(\theta))$ are equivalent. This proves the claim with Theorem \ref{complexity}.
						\end{proof}
						
						In the literature, the controllability and stabilizability radii in the real field and in the complex field have very different characterizations \cite{R.E1984Between,Hu2004Real}. Nevertheless, as a direct corollary of Theorem \ref{complexity} and Corollary \ref{complexity-stab}, we could still answer the complexity of the aforementioned problems in the complex field with the general affine structure.
						\begin{corollary} \label{complexity-stab-complex}
							Given $(A(\theta),B(\theta))$ in (\ref{affine}), checking the feasibility of the RSCR problem and the RSSZR problem are both NP-hard, even when $\theta$ can take complex values (i.e., $\theta\in {\mathbb C}^p$). 
						\end{corollary}
						
						\begin{proof}
							We can change $\theta\in {\mathbb{R}}^p$ to $\theta\in {\mathbb{C}}^p$ so that all the remaining statements in the proof of Theorem \ref{complexity} are still valid.
						\end{proof}
						
						It is notable that the hardness results in
						Theorem \ref{complexity} and Corollaries \ref{complexity-stab}-\ref{complexity-stab-complex} are irrespective of what form ${\bf \Gamma}(\theta)$ and matrix norm are utilized in the definitions. These results may explain why there is a lack of an efficient algorithm that could find the global optima of the RSCR and RSSZR problems (or even check the feasibility). 
						
						The NP-hardness of the addressed RSSR with respect to the $2$-norm can be obtained from a known result in robust stability analysis, shown as follows. However, it seems the complexity of the F-norm based RSSR still needs further investigation. 
						\begin{proposition}
							Given $A(\theta)$ in (\ref{affine}), the $2$-norm based RSSR with ${\bf \Gamma}(\theta)$ in Form (ii) is NP-hard.
						\end{proposition}
						
						\begin{proof}
							Proposition 2.1 of \cite{nemirovskii1993several} shows that it is NP-hard to decide whether all representatives of an interval matrix (i.e., a matrix with partially fixed constant entries and the rest in the interval $[-1,1]$) are stable. Let $A(\theta)$ be the affine parameterization of the aforementioned interval matrix. It then follows that, $r^{S}_{stb}(A(\theta))>1$ with ${\bf \Gamma}(\theta)$ of Form (ii), if and only if all representatives of $A(\theta)$ with $||\theta||_\infty\le 1$ are stable. Hence, deciding whether $r^{S}_{stb}(A(\theta))>1$ with ${\bf \Gamma}(\theta)$ of Form (ii) is NP-hard, leading to the required statement.
						\end{proof}

						Note provided $A$ is stable, due to the continuous dependence of eigenvalues on the matrix entries, the RSSR problem can be formulated as \cite{qiu1993formula}
						\begin{equation} \label{stable-formula}
							r_{stb}(A(\theta)) \!=\! \min \{ \left\| {\bf \Gamma}(\theta) \right\|: \theta \in {{\mathbb R}^{p}},  \exists w\in {\mathbb R}\ s.t.\  {\bf j}w\in {\bf \Lambda}(A(\theta))\},
						\end{equation}
						where $\bf j$ is the imaginary unit, recalling ${\bf \Lambda}(\cdot)$ is the set of eigenvalues.

						{ In the following, we prove the complexity of a problem that is closely related to the RSSR problem, namely, the problem of determining whether there is a $\theta\in {\mathbb R}^p$ such that $A(\theta)$ has a prescribed eigenvalue in the imaginary axis.  Similar to the proof of Theorem \ref{complexity}, we again resort to a reduction from the subset sum problem to an instance of the aforementioned problem. Given a set $S$ of $p$ integers $S=\{a_1,...,a_p\}$, our key idea is to construct a matrix $A(\theta)$ in such a way that $\det A(\theta)=\theta_1^2(\theta_1-1)^2+\cdots+\theta_p^2(\theta_p-1)^2+ (1+\sum \limits\nolimits_{i=1}^p \theta_ia_i)^2$.
						As a result, $\det A(\theta)=0$ if and only if $\theta_i\in \{0,1\}$ and $\sum \limits\nolimits_{i=1}^p \theta_ia_i=-1$, i.e., the subset sum problem on $S$ with target $T=-1$ is feasible. The construction of the structure of $A(\theta)$ is inspired by \cite{valiant1979completeness}. The details are given in the proof for the following result.} 

						\begin{theorem}\label{complexity-stable} Given $A(\theta)$ in (\ref{affine}), it is NP-hard to determine whether there exists a $\theta\in {\mathbb R}^p$ such that $A(\theta)$ has a prescribed eigenvalue in the imaginary axis.   
						\end{theorem}
						\begin{proof} We again give a reduction from the subset sum problem. Let $S=\{a_1,...,a_p\}$ be a set of $p$ integers and the target sum $T=-1$. Let $\theta=[\theta_1,...,\theta_p]^\intercal$. Inspired by \cite{valiant1979completeness}, which establishes the correspondence between polynomials and the matrix determinants, construct $A(\theta)\in {\mathbb R}^{(7p+9)\times (7p+9)}$ as the weighted adjacency matrix of the directed graph $G(\theta)$ illustrated in Fig. \ref{graph-representation}, that is, $[A(\theta)]_{ij}$ is the weight of the edge from node $i$ to node $j$ if it exists, and $[A(\theta)]_{ij}=0$ otherwise. Remember that all nodes except the one labeled `S' have self-loops with weight $1$.
							\begin{figure}[h]
								\centering
								\includegraphics[width=3.5in]{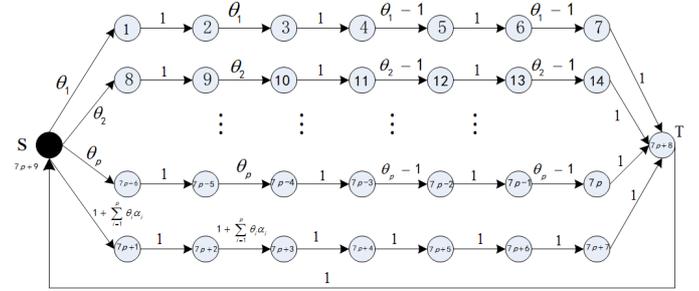}\\
								\caption{Graph representation $G(\theta)$ of $A(\theta)$. The number within each cycle is the node index, and on the edge is the weight of this edge. {\emph{Except for node $7p+9$ (labeled `S'), all the nodes have self-loops with weight $1$.}}}.\label{auc1}
								\label{graph-representation}
							\end{figure}
							
							Let $C$ be the set of cycle covers of $G(\theta)$, where a cycle cover is the set of edge-disjoint cycles that span the whole graph (i.e., visit each vertex of $G(\theta)$ exactly once). By definition of determinant \cite{R.A.1991Topics}, we have
							$$\det A(\theta)=\sum \limits_{c_i\in C} {\rm sign}(c_i) \prod \limits_{e \in E(c_i)} w(e), $$
							where $w(e)$ is the weight of edge $e$, $E(c_i)$ is the edge set of $c_i$, and ${\rm sign}(c_i)\in \{-1, 1\}$ is defined as \cite{R.A.1991Topics}
							$${\rm sign}(c_i)=\prod \limits_{\pi {\rm \ is \ a \ cycle \ of} \ c_i}(-1)^{{\rm length}(\pi)+1},$$
							in which ${\rm length}(\pi)$ is the number of edges in $\pi$. Observe that any cycle cover of $G(\theta)$ must contain a path from $S$ to $T$ and a path from $T$ to $S$, as well as $7p$ self-loops. Therefore, it is easy to get
							$$\det A(\theta)=\theta_1^2(\theta_1-1)^2+\cdots+\theta_p^2(\theta_p-1)^2+ (1+\sum \limits_{i=1}^p \theta_ia_i)^2.$$Obviously, $\det A({\bf 0}_p)=1\ne 0$.
							It is claimed that there exists a $\theta\in {\mathbb R}^p$ such that $A(\theta)$ has an eigenvalue ${\bf j}0$, if and only if the subset sum problem on $S$ is feasible.
							
							The only if direction is obvious, since to make $\det A(\theta)=0$ for a $\theta\in {\mathbb R}^p$, $\theta_i(\theta_i-1)=0, \forall i=1,...,p$ and $\sum \nolimits_{i=1}^p \theta_ia_i=-1$ should be satisfied. On the other hand, supposing $\theta\in \{0,1\}^p$ corresponds to a solution for the subset sum problem, we have $\sum \nolimits_{i=1}^p \theta_ia_i=-1$, which immediately leads to $\det A(\theta)=0$.
							
							Since $A(\theta)$ can be constructed in polynomial time, we conclude that deciding whether there is a $\theta\in {\mathbb R}^n$ such that $A(\theta)$ has a zero eigenvalue is NP-hard (and NP-complete). This leads to the required statement in Theorem \ref{complexity-stable}.
						\end{proof}

						\begin{remark} The problem in Theorem \ref{complexity-stable} is relevant in extending some existing methods for the RSSR problem to the general affine case in the literature.\footnote{It is worth mentioning that, in our proof for Theorem \ref{complexity-stable}, it is unclear whether $A(\theta)$ over $\theta \in {\mathbb R}^p$ contains a stable matrix.}
							For example, \cite{qiu1993formula,qiu1995computation,sreedhar1996fast,freitag2014new} proposed methods to find the global minimizer for $r^S_{stb}$, in which the perturbation is of the form $A_\Delta=F\Delta H$ for known $F, H$ and unknown full $\Delta$, and ${\bf \Gamma}(\theta)=\Delta$. In their methods, they first determine a frequency range for $w$, which is an interval of the form $[-w_{\min},w_{\max}]$ that contains the global minimizer, and then at each sample $w_s$ in this range (obtained via the grid or golden section search \cite{qiu1993formula,qiu1995computation,sreedhar1996fast}), determine the smallest $||\Delta||_2$ (denoted by $r(w_s)$) such that $A+A_\Delta-{\bf j}w_sI$ is singular, and finally find the smallest value of $r(w_s)$ among those samples, which is the global minimum of $r^S_{stb}$. The range $[-w_{\min},w_{\max}]$ is the whole $(-\infty,\infty)$ in \cite{qiu1993formula}, and analytically or iteratively determined in \cite{qiu1995computation, sreedhar1996fast,freitag2014new}, and determining $r({w_s})$ for a given $w_s$ is via a unimodal $\mu$-value related formula established in \cite{qiu1993formula}. Moreover, in the work of Katewa et al. \cite{katewa2020real}, the singularity of $A+A_\Delta$ (i.e., ${\bf j}w=0$)  is dealt with as a separate case from ${\bf j}w\ne 0$.  
						\end{remark}

						\begin{remark}
							Poljak et al. \cite{poljak1993checking} showed it is NP-hard to check the non-singularity of all representatives of an affinely parameterized matrix when the parameters are within the interval $[0,1]$ (termed robust non-singularity therein). Theorem \ref{complexity-stable} is an extension of that result by removing the interval constraints.
						\end{remark}
						
						\begin{remark}Unlike Corollary \ref{complexity-stab-complex}, Theorem \ref{complexity-stable} may not be valid if $\theta\in {\mathbb C}^p$.
							Note the NP-hardness of the problem in Theorem \ref{complexity-stable} does not imply that checking the {\emph{feasibility}} of the RSSR problem is NP-hard. Neither problem can be straightforwardly reduced to the other. 
						\end{remark}

						\section{Unified rank-relaxation based methods} \label{algorithms}
						Since we have established that determining the RSCR and RSSZR are both NP-hard, it is impossible to find the global optima in polynomial time unless $P=NP$.  In this section, towards the RSCR, RSSZR, and RSSR problems, we present unified rank-relaxation based algorithms with convergence guarantees. Here, `unified' refers to the two notable features of the proposed algorithms:  first, they share the same framework for the three different problems; second, for each problem, the respective algorithm is suitable for both the F-norm and $2$-norm based problems.
						
						We first present a useful lemma for the rank deficiency of complex-valued matrices.
						
						\begin{lemma}\label{real-rank}
							Given $X,Y\in {\mathbb R}^{n_1\times n_2}$, $X+{\bf j}Y$ is not of full row rank, if and only if $W\doteq \left[
							\begin{array}{cc}
								X & Y \\
								-Y & X \\
							\end{array}
							\right]
							$ is not of full row rank. {Moreover, $\sigma_{\min}(X+{\bf j}Y)=\sigma_{\min}(W)$.}
						\end{lemma}
					
				       \begin{proof}
				       	Suppose there is a nonzero row vector $x+{\bf j}y$ with $x,y\in {\mathbb R}^{1\times n_1}$ such that $(x+{\bf j}y)(X+{\bf j}Y)=0$. This yields
				       	$(xX-yY)+{\bf j}(yX+xY)=0$, which leads to $xX-yY=0$ and $yX+xY=0$. That is, $[x,y]W=0$, implying $W$ is not of full row rank. It is easy to see that the reverse direction also holds.
				       	
				       	Next, using the definition of singular value, we can write $\sigma_{\min}(X+{\bf j}Y)=\lambda^{\frac{1}{2}}_{\min}((X+{\bf j}Y)(X^\intercal-{\bf j}Y^\intercal))$, where
				       	$\lambda_{\min}(\cdot)$ takes the minimum eigenvalue, and $X^\intercal-{\bf j}Y^\intercal$ is the conjugate transpose of $X+{\bf j}Y$. Since $(X+{\bf j}Y)(X^\intercal-{\bf j}Y^\intercal)$ is Hermitian, it has real eigenvalues only. Let $\lambda\in {\mathbb R}$ be an eigenvalue of $(X+{\bf j}Y)(X^\intercal-{\bf j}Y^\intercal)$. Then, $(X+{\bf j}Y)(X^\intercal-{\bf j}Y^\intercal)-\lambda I=(XX^\intercal+YY^\intercal -\lambda I)+ {\bf j}(YX^\intercal-XY^\intercal)$ is of row rank deficient. By the first statement of Lemma \ref{real-rank}, we have that
				       	$$\left[
				       	\begin{array}{cc}
				       	XX^\intercal+YY^\intercal -\lambda I & YX^\intercal-XY^\intercal \\
				       	-YX^\intercal+XY^\intercal & XX^\intercal+YY^\intercal -\lambda I \\
				       	\end{array}
				       	\right]=WW^\intercal-\lambda I$$
				       	is of row rank deficient. Therefore, $\lambda$ is also an eigenvalue of $WW^\intercal$. The reverse direction also holds, that is, any $\lambda\in {\mathbb R}$ that is an eigenvalue of $WW^\intercal$ must also be an eigenvalue of $(X+{\bf j}Y)(X^\intercal-{\bf j}Y^\intercal)$.  Hence, $\sigma_{\min}(X+{\bf j}Y)=\lambda^{\frac{1}{2}}_{\min}((X+{\bf j}Y)(X^\intercal-{\bf j}Y^\intercal))=\lambda^{\frac{1}{2}}_{\min}(WW^\intercal)=\sigma_{\min}(W)$.
				       \end{proof}

						\begin{definition}[\cite{qi1996extreme}]\label{Ky-Fan-norm}
							Given a matrix $M\in {\mathbb F}^{n_1\times n_2}$ with ordered singular values $\sigma_1(M)\ge \cdots \ge \sigma_{\min\{n_1,n_2\}}(M)$, where ${\mathbb F}={\mathbb C}$ or $\mathbb R$, the nuclear norm of $M$ is the sum of all its singular values, i.e., $||M||_*=\sum \nolimits_{i=1}^{\min\{n_1,n_2\}} \sigma_i(M)$, and the Ky Fan $r$-norm of $M$ is the sum of the largest $r$ singular values, i.e., $||M||_{F_r}= \sum \nolimits_{i=1}^{r} \sigma_i(M)$.
						\end{definition}

						It is known that both $||M||_*$ and $||M||_{F_r}$ are convex in ${\mathbb F}^{n_1\times n_2}$ \cite{boyd2004convex}. In the following, define the function $$g(\theta)=||{\bf \Gamma}(\theta)||^2,$$where $||\cdot||$ is either the F-norm or the $2$-norm. Since ${\bf \Gamma}(\theta)$ is affine in $\theta$, and $||M||_F^2$ as well as $||M||_2^2$ is convex for $M$  in ${\mathbb F}^{n_1\times n_2}$, we conclude that $g(\theta)$ is convex for $\theta$ in ${\mathbb F}^{p}$ (${\mathbb F}={\mathbb C}$ or $\mathbb R$).\footnote{The convexity of $f(M)\doteq ||M||_2^2$ comes from the fact that, $f(M)$ is a composition function of $f=t^2, t\in [0,+\infty)$ and $t=||M||_2$, $M\in {\mathbb F}^{n_1\times n_2}$ which satisfies: (1) $t^2$ and $||M||_2$ are both convex in the corresponding domains; (2) $t^2$ is non-decreasing in $[0,+\infty)$; (3) The domain of $f(M)$ is convex (see \cite[Chap. 3.2.4]{boyd2004convex}). For ease of programming implementation, we may also use $g(\theta)=||{\bf \Gamma}(\theta)||_2$ for the $2$-norm instead of $g(\theta)=||{\bf \Gamma}(\theta)||^2_2$.}

						\subsection{Rank-relaxation based method for RSCR}
						According to the PBH test, the RSCR is equivalent to the following problem
						\begin{equation} \label{rank-PBH}
							\begin{aligned}
								& \mathop {\min } \limits_{\theta\in {\mathbb R}^p, \mu, \lambda \in {\mathbb R}} \ g(\theta)\\
								& {\rm s.t.} \ \ \   {\rm rank}[A(\theta)-({\bf j}\lambda + \mu)I, B(\theta)]<n
							\end{aligned}
						\end{equation}
						By Lemma \ref{real-rank}, we can remove the imaginary unit $\bf j$ such that problem (\ref{rank-PBH}) can be formulated as
						\begin{align} \label{rank-constraint-real}
							& \mathop {\min } \limits_{\theta\in {\mathbb R}^p, \mu, \lambda \in {\mathbb R}} \ g(\theta)\\
							&  {\rm s.t.} \ \ \     {\rm rank}\left[
							\begin{array}{cccc}
								A(\theta)-\mu I & B(\theta) & -\lambda I & 0 \\
								\lambda I & 0 & A(\theta)-\mu I & B(\theta) \\
							\end{array}
							\right]<2n, \label{rank-constraint}
						\end{align}
						which is a minimization problem over the rank constraint in the real field.
						
						The rank constraint is non-convex, making problem (\ref{rank-constraint-real}) sill hard to optimize. To handle this, inspired by \cite{hu2012fast}, we adopt the {\emph{truncated nuclear norm}} to replace the rank constraint. By Definition \ref{Ky-Fan-norm}, the constraint (\ref{rank-constraint}) is equivalent to
						\begin{equation}\label{trancated-norm}\sigma_{2n}(Z)=||Z||_*-||Z||_{F_{2n-1}}=0,\end{equation}
						where $Z\doteq \left[
						\begin{array}{cccc}
							A(\theta)-\mu I & B(\theta) & -\lambda I & 0 \\
							\lambda I & 0 & A(\theta)-\mu I & B(\theta) \\
						\end{array}
						\right]$. The sandwiched term of (\ref{trancated-norm}) is called the truncated nuclear norm (TNNR) in \cite{hu2012fast}. We, therefore, have a rank-relaxation formulation of problem (\ref{rank-PBH}):
						\begin{align} \label{rank-relaxation}
							& \mathop {\min } \limits_{\theta\in {\mathbb R}^p, \mu, \lambda \in {\mathbb R}, Z\in {\mathbb R}^{2n\times (2n+2m)}}  \ g(\theta)+\gamma(||Z||_*-||Z||_{F_{2n-1}})\\
							& {\rm s.t.} \ \ \   Z= \left[
							\begin{array}{cccc}
								A(\theta)-\mu I & B(\theta) & -\lambda I & 0 \\
								\lambda I & 0 & A(\theta)-\mu I & B(\theta) \\
							\end{array}
							\right]   \label{Z-equality}
						\end{align}
						where $\gamma>0$ is the regularization parameter, penalizing the smallest singular value of $Z$, which can be very large.  It is worth mentioning that, problem (\ref{rank-relaxation}) is just a relaxation of the RSCR. The optimal solution to this problem does not necessarily correspond to an optimal to the RSCR, and vice versa. {\emph{ It is marked that the introduction of the new variable $Z$ is to simplify the computation of (sub)differentials.}}

						As $(\theta, \mu,\lambda,Z)$ subject to (\ref{Z-equality}) is uniquely determined by $Z$,  we denote the  objective of problem (\ref{rank-relaxation}) at $(\theta, \mu,\lambda,Z)$ as $F(Z)$, i.e., $F(Z)=g(\theta)+\gamma(||Z||_*-||Z||_{F_{2n-1}})$. Note first the variable $Z$ depends on $\theta$, $\mu$ and $\lambda$ affinely in (\ref{Z-equality}), making the constraint set on $(\theta, \mu,\lambda,Z)$ convex. Additionally, the two terms $g(\theta)+\gamma||Z||_*$ and $\gamma||Z||_{F_{2n-1}}$ in the objective are both convex in $Z$ and $\theta$.  Therefore, problem (\ref{rank-relaxation}) is a difference-of-convex (DC) program \cite{yuille2003concave}. For DC programs, we shall adopt the sequential convex relaxations \cite{yuille2003concave} to find a stationary point of problem (\ref{rank-relaxation}). Recall a feasible solution is said to be a stationary point of an optimization problem if it satisfies the corresponding Karush-Kuhn-Tucker (KKT) conditions, which are necessary for the local optimality (see \citep[Sec 2]{razaviyayn2013unified}).
						
						To better describe the sequential convex relaxation procedure, let  $(\theta^{(k)}, \mu^{(k)},\lambda^{(k)},Z^{(k)})$ be the point obtained in the $k$th iteration. To develop this procedure, the key is to determine the first-order derivative of the second convex term (the concave term as an addend) in the objective, such that we can linearize it around $(\theta^{(k)}, \mu^{(k)},\lambda^{(k)},Z^{(k)})$. To this end, let the singular value decomposition (SVD) of $Z^{(k)}$ be
						\begin{equation}\label{svd} Z^{(k)}=[U_1^{(k)}, U_2^{(k)}]\left[
							\begin{array}{cc}
								\Lambda_1 &  \\
								& \Lambda_2 \\
							\end{array}
							\right]\left[
							\begin{array}{c}
								V^{(k),\intercal}_1 \\
								V^{(k),\intercal}_2 \\
							\end{array}
							\right]
							,\end{equation} where $U_1^{(k)}$ and $V_1^{(k)}$ are respectively the left and right singular vectors associated with the largest $2n-1$ singular values of $Z^{(k)}$. From \citep[Theorem 3.4]{qi1996extreme}, $U_1^{(k)}V_1^{(k),\intercal}\in  \partial ||Z||_{F_{2n-1}}$ ($\partial$ denotes the subdifferential). Then, $||Z||_{F_{2n-1}}$ can be linearized at $Z^{(k)}$ as \cite{watson1992characterization,qi1996extreme}
						\begin{equation}\label{Fr-norm}||Z||_{F_{2n-1}}\thickapprox ||Z^{(k)}||_{F_{2n-1}}+{\rm tr}(U_1^{(k),\intercal}(Z-Z^{(k)})V_1^{(k)}),\end{equation}
						and from the convexity of $||Z||_{F_{2n-1}}$ we establish
						\begin{equation} \label{convexity}
							\begin{array}{l}||Z||_{F_{2n-1}}\ge ||Z^{(k)}||_{F_{2n-1}}+{\rm tr}(U_1^{(k),\intercal}(Z-Z^{(k)})V_1^{(k)}),\\ \forall Z, Z^{(k)}\in {\mathbb R}^{2n\times (2n+2m)},\end{array}
						\end{equation}
						where ${\rm tr}(\cdot)$ takes the trace.

						Based on the above expressions, the program at the $(k+1)$th iteration is formulated as
						\begin{equation} \label{sub-convex-iterate}
							\begin{aligned}
								\mathop {\min }\limits_{\theta, \mu, \lambda, Z} &{\kern 3pt} g(\theta) + \gamma||Z||_{*}-\gamma {\rm tr}(U_1^{(k),\intercal}ZV_1^{(k)})\\
								s.t. \ \ \  & Z= \left[
								\begin{array}{cccc}
									A(\theta)-\mu I & B(\theta) & -\lambda I & 0 \\
									\lambda I & 0 & A(\theta)-\mu I & B(\theta) \\
								\end{array}
								\right]
							\end{aligned}
						\end{equation}
						Problem (\ref{sub-convex-iterate}) is convex since its objective function as well as constraint set is convex, and thus can be solved efficiently using the off-the-shelf solvers for convex optimizations. The complete iterative procedure is collected in Algorithm \ref{alg2}, recalling that $F(Z^{(k)})=g(\theta^{(k) })+\gamma(||Z^{(k)}||_*-||Z^{(k)}||_{F_{2n-1}})$. The convergence of Algorithm \ref{alg2} is stated in the following theorem.

						\begin{algorithm}[H] 
							{{{
										\caption{: A rank-relaxation algorithm for RSCR}
										\label{alg2} 
										\begin{algorithmic}[1]
											\STATE {Set $k=0$ and randomly initialize  $\theta^{(0)}, \lambda^{(0)}, \mu^{(0)}$;  }
											\WHILE {$|F(Z^{(k)})-F(Z^{(k-1)})|>\xi$ ($\xi>0$ is the convergence threshold; assume the procedure is executed when $k=0$)} 
											\STATE Obtain $U_1^{(k)}$ and $V_1^{(k)}$ according to (\ref{svd}) via the SVD on $Z^{(k)}$;
											\STATE Solve the convex program (\ref{sub-convex-iterate}) to obtain $(\theta^{(k+1)},\lambda^{(k+1)}, \mu^{(k+1)}, Z^{(k+1)})$;
											\STATE $k+1 \leftarrow k$;
											\ENDWHILE
											\STATE Return $\theta^{(k)}$ when convergence.
								\end{algorithmic}}}
							}
						\end{algorithm}
					
											\begin{theorem}\label{converge} The sequence $\{(\theta^{(k)}, \mu^{(k)},\lambda^{(k)},Z^{(k)})\}$ generated by Algorithm \ref{alg2} satisfies:
						\begin{itemize}
							\item[(i)] $F(Z^{(k+1)})\le F(Z^{(k)})$, $k=0,1,\cdots,$\\~ $\lim \limits_{k\rightarrow \infty} (F(Z^{(k+1)})-F(Z^{(k)}))=0$;
							\item[(ii)] $\lim \limits_{k\rightarrow \infty} (\theta^{(k+1)}-\theta^{(k)})=0$.
						\end{itemize}
						Moreover, Algorithm \ref{alg2} is guaranteed to converge to a stationary point of problem (\ref{rank-relaxation}).
					\end{theorem}
					
					\begin{proof} The non-increasing of $F(Z^{(k)})$ can be proved in a  way similar to \citep[Theorem 2]{yuille2003concave}. To be specific, let the function $G(Z,Z^{(k)})\doteq g(\theta)+\gamma||Z||_*-\gamma(||Z^{(k)}||_{F_{2n-1}}+{\rm tr}(U_1^{(k),\intercal}(Z-Z^{(k)})V_1^{(k)}))$. Then, from (\ref{convexity}), we have
						\begin{equation} \label{upper-bound}
						\begin{array}{l}	F(Z)\le G(Z,Z^{(k)}), F(Z^{(k)})=G(Z^{(k)},Z^{(k)}), \\ \\ \forall Z,\ Z^{(k)}\in {\mathbb{R}}^{2n\times (2n+2m)}. \end{array}
						\end{equation}
						We therefore have
						$$F(Z^{(k+1)})\le G(Z^{(k+1)},Z^{(k)})\le G(Z^{(k)},Z^{(k)})=F(Z^{(k)}),$$
						where the first inequality and the last one come from (\ref{upper-bound}) and the sandwiched one from the optimization on (\ref{sub-convex-iterate}).
						
						Since $F(Z^{(k)})\ge 0$ is lower bounded, we obtain that the sequence $\{F(Z^{(k)})\}$ has a limit. The existence of a limit for $\{\theta^{(k)}\}$ then follows from the continuity of $F(Z)$ (so that the level set $\{Z\ |F(Z)\le F(Z^{(k)})\}$ is compact for any $Z^{(k)}$).  For the rest of statement (ii), note Algorithm \ref{alg2} is in fact a majorization-minimization (MM) algorithm (see \citep[Section 2]{lanckriet2009convergence} for details). It is easy to verify that the approximation function $G(Z,Z^{(k)})$ of $F(Z)$
						satisfies the sufficient conditions in \citep[Theorem 1 \& Corollary 1]{razaviyayn2013unified} for the convergence of the MM algorithm. According to that, Algorithm \ref{alg2} is guaranteed to converge to a stationary point of  problem (\ref{rank-relaxation}).
					\end{proof}	
						
						{\small
							\begin{algorithm}[H] 
								{{{
											\caption{: A two-stage rank-relaxation algorithm for RSCR with $\epsilon$-tolerance}
											\label{alg2-weighting} 
											\begin{algorithmic}[1]
											
												\STATE {Set $k=0$ and  initialize $\theta^{(0)}, \lambda^{(0)}, \mu^{(0)}$ and $Z^{(0)}$ by solving the DC program using the sequential convex program (i.e., the similar algorithm to Algorithm \ref{alg2} by changing the objective of (\ref{sub-convex-iterate}) to $||Z||_{*}-{\rm tr}(U_1^{(k),\intercal}ZV_1^{(k)})$ with a random initial point $(\tilde \theta^{(0)}, \tilde \lambda^{(0)}, \tilde \mu^{(0)})$)
													\begin{equation}  \label{initial-prob-RSCR}
														\begin{aligned}
															& \mathop {\min }\limits_{\theta, \mu, \lambda, Z} ||Z||_{*}-||Z||_{F_{2n-1}} \\
															&	s.t. \ \ (\ref{Z-equality})
														\end{aligned}
													\end{equation}
												}
												//$\rightarrow$[First stage]
												\WHILE {$|F(Z^{(k)})-F(Z^{(k-1)})|>\xi$ ($\xi>0$ is the convergence threshold; assume the procedure is executed when $k=0$)} 
												\STATE Obtain $U_1^{(k)}$ and $V_1^{(k)}$ according to (\ref{svd}) via the SVD on $Z^{(k)}$;
												\STATE Set $\gamma= \frac{g(\theta^{(0)})}{\epsilon}$. Solve the convex program (\ref{sub-convex-iterate}) to obtain $(\theta^{(k+1)},\lambda^{(k+1)}, \mu^{(k+1)}, Z^{(k+1)})$;
												\STATE $k+1 \leftarrow k$;
												\ENDWHILE~\\
												 //$\rightarrow$[Second stage]
												\STATE Return $\theta^{(k)}$ when convergence.
									\end{algorithmic}}}
								}
							\end{algorithm}
						}

						
						It is noted that there may be multiple stationary points for the relaxed problem (\ref{rank-relaxation}). In this case, the local minima returned by Algorithm \ref{alg2} would depend on the initial points \cite{razaviyayn2013unified} (this property is also shared by other iterative algorithms for the RSCR \cite{johnson2018structured}).   Moreover, for a returned local optimal solution $(\theta^*,\lambda^*,\mu^*,Z^*)$ of problem (\ref{rank-relaxation}), there is usually no guarantee that it is a feasible solution for the original problem (\ref{rank-constraint-real}), since $\sigma_{2n}(Z^*)$ may not be zero.
						
					
					{
			To ensure that the obtained solutions meet certain quality specifications, we introduce an initialization stage to Algorithm \ref{alg2} that leads to a two-stage algorithm for RSCR with $\epsilon$-tolerance (see Algorithm \ref{alg2-weighting}). Specifically, for any $\epsilon>0$, we aim to obtain a solution $(\theta^*,\lambda^*,\mu^*,Z^*)$ that satisfies the $\epsilon$-tolerance condition $\sigma_{2n}(Z^*)\le \epsilon$. By Lemma \ref{real-rank}, this guarantees that
	  	  \begin{equation}\label{quality} \sigma_{\min}([A(\theta^*)-({\bf j}\lambda^*+\mu^*)I, B(\theta^*)])=\sigma_{\min}(Z^*)\le \epsilon.\end{equation}When $\epsilon$ is small enough, the obtained system $(A(\theta^*), B(\theta^*))$ could be regarded as uncontrollability.
			 In the first stage of the algorithm, our goal is to determine a feasible (or near-feasible, given that verifying the feasibility of the RSCR is NP-complete) solution to the RSCR problem. To this end, we use problem (\ref{initial-prob-RSCR}) as the initial condition. The reason for this choice is that, given that the original RSCR problem is feasible, the global optimum of problem (\ref{initial-prob-RSCR}) is exactly $0$. Moreover, any feasible solution to the RSCR problem is also a global optimal solution to problem (\ref{initial-prob-RSCR}). However, if a feasible solution $\theta^{(0)}$ to the original RSCR problem is already available, the first stage of Algorithm \ref{alg2-weighting} is not necessary. Note that this stage does not optimize the original objective function $g(\theta)$.
			
			In the second stage of the algorithm, we choose the regularization parameter $\gamma = \frac{g(\theta^{(0)}}{\epsilon}$ due to the non-decreasing nature of $F(Z^{(k)})$. We provide a detailed explanation for this choice in the proof of the following theorem.		
				
			}
						
%
%

						
						\begin{theorem} \label{convergence-epsilon}
							Assume the initial point \\ $(\theta^{(0)}, \lambda^{(0)}, \mu^{(0)}, Z^{(0)})$ attained in the first stage of Algorithm \ref{alg2-weighting} satisfies the uncontrollability constraint (\ref{rank-constraint}). Let the sequence $\{(\theta^{(k)}, \mu^{(k)},\lambda^{(k)},Z^{(k)})\}$ be generated by Algorithm \ref{alg2-weighting}. Then, we have
						{\begin{itemize}
								\item[(i)] $F(Z^{(k+1)})\le F(Z^{(k)})$, $k=0,1,\cdots,$\\~ $\lim \limits_{k\rightarrow \infty} (F(Z^{(k+1)})-F(Z^{(k)}))=0$;
								\item[(ii)] $\lim \limits_{k\rightarrow \infty} (\Upsilon^{(k+1)}-\Upsilon^{(k)})=0$, where $\Upsilon=\theta,\lambda,\mu$ and $Z$;
								\item[(iii)] Let $(\theta^*,\lambda^*,\mu^*,Z^*)$ be the limit of $(\theta^{(k)},\lambda^{(k)},\mu^{(k)},Z^{(k)})$ as $k\to \infty$. Then, $(\theta^*,\lambda^*,\mu^*,Z^*)$ is a stationary point of problem (\ref{rank-relaxation}) satisfying
								                    $$\sigma_{\min}([A(\theta^*)-({\bf j}\lambda^*+\mu^*)I, B(\theta^*)])\le \epsilon.$$
							\end{itemize}}
%
						\end{theorem}
						\begin{proof} Statements (i) and (ii) follow similar arguments to the proof of Theorem \ref{converge}. Hence, we only need to show $\sigma_{2n}(Z^*)\le \epsilon$.
							{	By the non-decreasing of $F(Z^{(k)})$ described in statement (i), we have $$F(Z^{*})=g(\theta^*)+\gamma\sigma_{2n}(Z^*)\le F(Z^{(0)})=g(\theta^{(0)}).$$ As $g(\theta^*)\ge 0$ and $\gamma = \frac{g(\theta^{(0)}}{\epsilon}$, we get $\sigma_{2n}(Z^*)\le \epsilon$. The required result then follows from Lemma \ref{real-rank}.}
						\end{proof}
						

	{	\begin{remark}
	 Although both Algorithm \ref{alg2} and Algorithm \ref{alg2-weighting} are guaranteed to converge to a local minimizer, our simulations in Section \ref{examples} indicate that Algorithm \ref{alg2-weighting} is generally capable of finding solutions with better uncontrollability quality (measured by the smallest $\epsilon$ satisfying (\ref{quality})) than Algorithm \ref{alg2}, given the same initial points. However, it is important to note that this improved performance comes at the cost of more computational resources needed to execute the first (initialization) stage.
       \end{remark} }

						The scheme for the RSCR could be extended to algorithms for the RSSZR and RSSR with appropriate modifications. In the sequel, we shall present these extensions without proof, due to their similarity to those for the RSCR.

						\subsection{Rank-relaxation based method for RSSZR}
						In line with the previous subsection, we briefly give the scheme of the rank-relaxation based method for RSSZR. By definition of stabilizability, the RSSZR problem can be formulated as problem (\ref{rank-PBH}) with the additional constraint $\mu\ge 0$. Thus, it can be rewritten as the following problem in the real field:
						\begin{align} \label{rank-constraint-real-stabiliability}
							&		\mathop {\min } \limits_{\theta\in {\mathbb R}^p, \mu, \lambda \in {\mathbb R}}  \ g(\theta)\\
							&		{\rm s.t.} \ \ \      {\rm rank}\left[
							\begin{array}{cccc}
								A(\theta)-\mu I & B(\theta) & -\lambda I & 0 \\
								\lambda I & 0 & A(\theta)-\mu I & B(\theta) \\
							\end{array}
							\right]<2n, \label{rank-constraint-stabilizable-2} \\
							\ \ &	\mu\ge 0. \label{positive-mode}
						\end{align}
						In the spirit of the TNNR, the rank-relaxation formulation of problem (\ref{rank-constraint-real-stabiliability}) is
						\begin{align} \label{rank-relaxation-stabiliability}
							\mathop {\min } \limits_{\theta\in {\mathbb R}^p, \mu, \lambda \in {\mathbb R}, Z\in {\mathbb R}^{2n\times (2n+2m)}} & \ g(\theta)+\gamma(||Z||_*-||Z||_{F_{2n-1}})\\
							{\rm s.t.} \ \ \ &  (\ref{Z-equality}) \ {\rm and} \ \mu\ge 0,
						\end{align}where $\gamma>0$ is the regularization parameter.
						
						Similar to problem (\ref{rank-relaxation}),  a practical way to find local optima of problem (\ref{rank-relaxation-stabiliability}) is via the sequential convex relaxation. Analogous to Algorithm \ref{alg2-weighting},  the procedure of the sequential convex relaxation with $\epsilon$-tolerance for `unstabilizability' is summarized in Algorithm \ref{alg2-RSSZR}, in which $F(Z^{(k)})$ is defined in the same way as that for the RSCR problem. Similar to Proposition \ref{convergence-epsilon}, Algorithm \ref{alg2-RSSZR} is guaranteed to converge to a stationary point $(\theta^{*}, \lambda^{*}, \mu^{*}, Z^{*})$ of problem (\ref{rank-relaxation-stabiliability}), and $\sigma_{2n}(Z^*)<\epsilon$ for a prescribed tolerance $\epsilon$, provided that $Z^{(0)}$ satisfies (\ref{rank-constraint-stabilizable-2})-(\ref{positive-mode}).
						
						\begin{algorithm}[H] 
							{{{
										\caption{: A two-stage rank-relaxation algorithm for RSSZR with $\epsilon$-tolerance} 
										\label{alg2-RSSZR} 
										\begin{algorithmic}[1]
											\STATE {Initialize $k=0$ and $\theta^{(0)}, \lambda^{(0)}, \mu^{(0)}$ and $Z^{(0)}$ by solving the DC program in a similar way to Algorithm \ref{alg2} with random initial conditions $(\tilde \theta^{(0)}, \tilde \lambda^{(0)}, \tilde \mu^{(0)})$
												\begin{equation}  \label{initial-prob}
													\begin{aligned}
														& \mathop {\min }\limits_{\theta, \mu, \lambda, Z} ||Z||_{*}-||Z||_{F_{2n-1}} \\
														s.t. \ & \ (\ref{Z-equality}) \ {\rm and} \ \mu\ge 0.
													\end{aligned}
												\end{equation}
											}
											\WHILE {$|F(Z^{(k)})-F(Z^{(k-1)})|>\xi$ ($\xi>0$ is the convergence threshold; assume the procedure is executed when $k=0$)} 
											\STATE Obtain $U_1^{(k)}$ and $V_1^{(k)}$ according to (\ref{svd}) via the SVD on $Z^{(k)}$;
											\STATE Set $\gamma= \frac{g(\theta^{(0)})}{\epsilon}$. Solve the following convex program to obtain $(\theta^{(k+1)},\lambda^{(k+1)}, \mu^{(k+1)}, Z^{(k+1)})$:
											\begin{equation} \label{sub-convex-iterate-RSSZR}
												\begin{aligned}
													\mathop {\min }\limits_{\theta, \mu, \lambda, Z} &{\kern 3pt} g(\theta) + \gamma||Z||_{*}-\gamma {\rm tr}(U_1^{(k),\intercal}ZV_1^{(k)})\\
													s.t. \ \ \  &  \ (\ref{Z-equality}) \ {\rm and} \ \mu\ge 0.
												\end{aligned}
											\end{equation}
											\STATE $k+1 \leftarrow k$;
											\ENDWHILE
											\STATE Return $\theta^{(k)}$ when convergence.
								\end{algorithmic}}}
							}
						\end{algorithm}
						
						\subsection{Rank-relaxation based method for RSSR}
						
						According to (\ref{stable-formula}), provided that $A$ is stable, the RSSR problem can be formulated as
						\begin{align} \label{rank-constraint-real-RSSR}
							\mathop {\min } \limits_{\theta\in {\mathbb R}^p, \lambda \in {\mathbb R}} & \ g(\theta)\\
							{\rm s.t.} \ \ \ &     {\rm rank}\left[
							\begin{array}{cc}
								A(\theta) & -\lambda I  \\
								\lambda I  & A(\theta) \\
							\end{array}
							\right]<2n, \label{rank-constraint-stable}
						\end{align}
						In the same line as (\ref{rank-relaxation}) and (\ref{rank-relaxation-stabiliability}),  we can obtain the rank-relaxation of problem (\ref{rank-constraint-real-RSSR}) as a DC program:
						\begin{align} \label{rank-relaxation-RSSR}
							\mathop {\min } \limits_{\theta\in {\mathbb R}^p, \lambda \in {\mathbb R}, Z\in {\mathbb R}^{2n\times 2n}} & \ F(Z)=g(\theta)+\gamma(||Z||_*-||Z||_{F_{2n-1}})\\
							{\rm s.t.} \ \ \ & Z=\left[
							\begin{array}{cc}
								A(\theta) & -\lambda I  \\
								\lambda I  & A(\theta) \\
							\end{array}
							\right]\label{Z-equality-RSSR}
						\end{align}with $\gamma>0$ being the regularization parameter.
						Similar to Algorithms \ref{alg2-weighting} and \ref{alg2-RSSZR}, we can develop a sequential-convex relaxation based algorithm to find stationary point solutions of problem (\ref{rank-relaxation-RSSR}) with $\epsilon$-tolerance for $\sigma_{2n}(Z^{(k)})$ to zero when convergence (provided that $Z^{(0)}$ satisfies constraint (\ref{rank-constraint-stable})), which is given as Algorithm \ref{alg2-RSSR}.
						
						\begin{algorithm} 
							{{{
										\caption{: A two-stage rank-relaxation algorithm for RSSR with $\epsilon$-tolerance (assuming $A$ stable)} 
										\label{alg2-RSSR} 
										\begin{algorithmic}[1]
											\STATE {Initialize $k=0$ and $\theta^{(0)}, \lambda^{(0)}$, and $Z^{(0)}$ by solving the DC program in a similar way to Algorithm \ref{alg2} with random initial conditions $(\tilde \theta^{(0)}, \tilde \lambda^{(0)})$
												\begin{equation*}
													\begin{aligned}
														& \mathop {\min }\limits_{\theta, \lambda, Z} ||Z||_{*}-||Z||_{F_{2n-1}} \\
														s.t. \ & \ (\ref{Z-equality-RSSR}).
													\end{aligned}
												\end{equation*}
											}
											\WHILE {$||F(\theta^{(k)})-F(\theta^{(k-1)})||>\xi$ (assume the procedure is executed when $k=0$)} 
											\STATE Obtain $U_1^{(k)}$ and $V_1^{(k)}$ according to (\ref{svd}) via the SVD on $Z^{(k)}$;
											\STATE Solve the following convex program to obtain $(\theta^{(k+1)},\lambda^{(k+1)}, Z^{(k+1)})$:
											\begin{equation} \label{convex-for-RSSR}
												\begin{aligned}
													\mathop {\min }\limits_{\theta, \lambda, Z} &{\kern 3pt} g(\theta) + \gamma||Z||_{*}-\gamma {\rm tr}(U_1^{(k),\intercal}ZV_1^{(k)})\\
													s.t. \ \ \  &  (\ref{Z-equality-RSSR}).
												\end{aligned}
											\end{equation}
											\STATE $k+1 \leftarrow k$;
											\ENDWHILE
											\STATE Return $\theta^{(k)}$ when convergence.
								\end{algorithmic}}}
							}
						\end{algorithm}
					
						\begin{remark} \label{sdp-solver}
						The convex programs (\ref{sub-convex-iterate}), (\ref{sub-convex-iterate-RSSZR}), and (\ref{convex-for-RSSR}) can all be solved by off-the-shelf solvers for semi-definite programs (SDP) such as SeDuMi, SDPT3, etc., by inspecting that, for $Z\in {\mathbb R}^{n_1\times n_2}$, $$||Z||_*=\min \limits_{{\tiny\begin{array}{c}
								W_1\in {\mathbb R}^{n_1\times n_1}\\ W_2\in {\mathbb R}^{n_2\times n_2}
								\end{array}}} \left\{\frac{1}{2}{\rm tr}(W_1+W_2): \left[
						\begin{array}{cc}
						W_1 & Z \\
						Z^\intercal & W_2 \\
						\end{array}
						\right]\succeq 0\right\}.
						$$
					\end{remark}
						
							
						\begin{remark}[Comparison with existing methods] 	To the best of our knowledge, the only available algorithms that can deal with the general affine structure are the structured total least squares-based algorithm proposed in \cite{khare2012computing} (see footnote 2), and the iterative algorithms presented in \cite{johnson2018structured}. Both of these algorithms transform the rank constraints into bilinear equalities by introducing extra variables that serve as the eigenvectors or singular vectors of certain matrices. For example, for the RSCR, the bilinear equality takes the form of $x^\intercal[A(\theta)-zI,B(\theta)]=0$, where $x\in {\mathbb C}^n$ and $z\in {\mathbb{C}}$, by introducing the variable $x$.
							In contrast, our proposed methods avoid introducing bilinear constraints (as well as the variable $x$) by relaxing the rank constraints with the regularized TNNR term, which is fundamentally different from the techniques used in \cite{khare2012computing,johnson2018structured,katewa2020real}. Furthermore, our proposed algorithms converge to local optima without relying on any intermediate assumptions. Even without the initialization stage, convergence is still guaranteed. In contrast, specific assumptions are required for the algorithms in \cite{johnson2018structured,katewa2020real}.				
						{	Another notable feature of our proposed methods is that, as mentioned earlier, it can handle $2$-norm and $F$-norm based problems in a unified way. In contrast, the methods in \cite{khare2012computing,johnson2018structured,katewa2020real} are all tailored for $F$-norm and are not directly applicable to $2$-norm based problems.}
							
						\end{remark}
					
					{ \begin{remark}[Limitations]
						Due to the NP-hardness of the considered problems, like \cite{johnson2018structured,bianchin2016observability,khare2012computing,katewa2020real}, our proposed algorithms could only find the local minima. The returned local minima may be influenced by the initial points of the iterative algorithms. Besides, even though Algorithms \ref{alg2-weighting}, \ref{alg2-RSSZR}, and \ref{alg2-RSSR} have given the expression of the regularization parameter $\gamma$, we found in practice that the computational cost can increase significantly when this value is too large. Therefore, how to choose a better $\gamma$ remains challenging. Some adaptive update strategy may be preferred~\cite{lin2011linearized}. 
					  \end{remark}}
					
						
					{	\begin{remark}[Computational complexity] In each iteration of our proposed algorithms, the computational complexity is dominated by the SVD on $Z^{(k)}$ with dimension $2n\times (2n+2m)$ for the RSCR and RSSZR, and $2n\times 2n$ for the RSSR, and solving the respective convex programs (\ref{sub-convex-iterate}), (\ref{sub-convex-iterate-RSSZR}), and (\ref{convex-for-RSSR}). The former has a time complexity of $O(n^2(n+m))$ ($O(n^3)$ for the RSSR) \cite{R.A.1991Topics}. The latter can be solved by some commercial SDP solvers in a reasonably fast time, including SeDuMi and SDPT3 mentioned in Remark \ref{sdp-solver}. The complexity of solving an SDP is upper bounded by $O(k^{5.246})$ with $k$ being the dimension of the involved semi-definite matrices when using the interior-point method in \cite{jiang2020faster}. Therefore, the complexity of each iteration in our methods is upper bounded by $O(n^{5.246})$. This complexity is comparable to the $O(n^5)$ complexity in each iteration of the algorithm given in \cite{johnson2018structured}, which is devoted to the $F$-norm based problems only. Moreover, by inspecting the inherent TNNR structure of (\ref{sub-convex-iterate}), (\ref{sub-convex-iterate-RSSZR}), and (\ref{convex-for-RSSR}), we can resort to the specialized algorithm named TNNR-APCL (TNNR accelerated proximal gradient line method) give in \cite{hu2012fast} to solve those convex programs more efficiently. In each iteration of TNNR-APCL, the main computational cost is the computation of SVD, and TNNR-APCL has a quadratic convergence rate. This may make our methods feasible for large-scale systems.  
							
						\end{remark}}
						

						\section{Nuerical Examples} \label{examples}
						This section provides several examples to demonstrate the effectiveness of the proposed methods.
						\begin{example}\label{stability-example} Consider the system that appeared in \cite{qiu1993formula} and \cite{katewa2020real} with {\small $A=\left[\begin{matrix}
									79&20&-30&-20\\
									-41&-12&17&13\\
									167&40&-60&-38\\
									33.5&9&-14.5&-11
								\end{matrix}\right]$} and the perturbation $A_\Delta=E\Delta H$, where
							{\small $E=\left[\begin{matrix}
									0.2190 & 0.9347 \\
									0.0470& 0.3835 \\
									0.6789& 0.5194 \\
									0.6793& 0.8310
								\end{matrix}\right]$,  $H=\left[\begin{array}{cccc}
									0.0346&0.5297&0.0077&0.0668 \\
									0.0535&0.6711&0.3848&0.4175
								\end{array}\right]$}
							and $\Delta\in {\mathbb R}^{2\times 2}$. As in \cite{katewa2020real}, we investigate the following two classes of structural constraints on $\Delta$, and for each one we shall consider the RSSR in terms of both the F-norm (i.e. $||\Delta||_F$) and the $2$-norm (i.e., $||\Delta||_2$):

								%

							\begin{itemize}
								\item Case I: $\Delta$ is full, corresponding to $p=4$;
								
								\item Case II:  $\Delta$ is diagonal, corresponding to $p=2$.
							\end{itemize}
							As mentioned earlier, in Case II, $||\Delta||_2$ is exactly the maximum magnitude of diagonal entries of $\Delta$.

							As $A$ is stable (its eigenvalues are $\{-1\pm {\bf j}10, -1\pm {\bf j}1\}$), we use Algorithm \ref{alg2-RSSR} to compute its RSSR. Particularly, $\gamma=\min\{5,g(\theta^{(0)})/\epsilon\}$ (due to the limitation on the computer precision, if $\gamma$ is too large then the term $g(\theta)$ might be neglected compared with $\gamma(||Z||_*-||Z||_{F_{2n-1}})$ in $F(Z)$; also, we found a larger $\gamma$ may take a relatively long time to converge), with $\epsilon=10^{-4}$. The convergence threshold is $\xi=10^{-5}$ for both stages of Algorithm \ref{alg2-RSSR}. The maximum number of iterations is $600$.   Since there may be multiple stationary points for the corresponding rank-relaxed problems, we conduct $100$ trials of Algorithm \ref{alg2-RSSR} for each structural constraint and each norm, in which for each trial we randomly choose the initial $\tilde \theta^{(0)}$ with each of its elements uniformly distributed in $(0,1)$ and $\tilde \lambda^{(0)}$ in $(-2,2)$.
							
							We collect the results in Table \ref{stability-case}. In this table, the successful rate comes from the following facts:  From \cite{qiu1993formula} and \cite{katewa2020real}, we know the optimal $r_{stb}$ w.r.t. the F-norm for Case I and Case II are respectively $0.5159$ and $0.5653$, and w.r.t. the $2$-norm for Case I is $0.5141$.\footnote{The weeny difference between $0.5141$ in \cite{qiu1993formula} and $0.5132$ obtained herein may result from the rounding errors and different computation precisions.} The optimality of $0.5284$ for Case II w.r.t. the $2$-norm can be verified from Fig \ref{optimality-verify}, which plots the regions of all eigenvalues of $A+E\Delta H$ with $\Delta={\bf diag}\{\theta_1,\theta_2\}$ and $-0.5284\le \theta_1,\theta_2\le 0.5284$. We present four typical successful solutions in Table \ref{typical-cases},  as well as their corresponding (randomly generated) initial conditions, perturbations $\Delta$, and eigenvalues of the resulting $A+A_{\Delta}$. For the four cases, the evolution of the relaxed objective function $F(Z^{(k)})$ for the second stage of Algorithm \ref{alg2-RSSR} during the iteration is given in
							Fig. \ref{iteration-function}. This figure shows the non-increasing property of $F(Z^{(k)})$.
							
							From Table \ref{stability-case}, it is found: (1) With very high probability, Algorithm \ref{alg2-RSSR} will return the optimal solutions for all four cases. In particular, a successful rate of $100\%$ has been obtained for Case II w.r.t. the F-norm. This indicates the proposed method is relatively robust to the initial conditions. (2) Determining the $2$-norm based RSSR tends to require more iterations for Algorithm \ref{alg2-RSSR} to converge than that of the F-norm based RSSR. This is perhaps because the F-norm based optimization is easier to handle compared with the $2$-norm based one in most off-the-shelf convex optimization solvers \cite{boyd2004convex} (possibly because the gradient of the F-norm is easier to compute).

							\begin{table*}
								\centering
								\begin{threeparttable}
									\caption{RSSR via Algorithm \ref{alg2-RSSR} in $100$ experiments} \label{stability-case}
									\begin{tabular}{c|c|c|c|c}
										\hline
										Cases                             &  RSSR ($r_{stb}$)     & Frequency & Successful rate       & Average iterations\tnote{1}     \\ \hline
										\multirow{2}{*}{Case I (2-norm)} & 0.5132 & 96          & \multirow{2}{*}{{\bf 96 }\%} & \multirow{2}{*}{34.37} \\ \cline{2-3}
										& 1.0718 & 4           &                       &                        \\ \hline
										\multirow{3}{*}{Case I (F-norm)} & 0.5159 & 96          & \multirow{3}{*}{{\bf 96 }\%} & \multirow{3}{*}{12.64} \\ \cline{2-3}
										& 1.0350 & 2           &                       &                        \\ \cline{2-3}
										& 1.0717 & 2           &                       &                        \\ \hline
										\multirow{2}{*}{Case II (2-norm)} & 0.5284 & 89          & \multirow{2}{*}{{\bf 89 }\%} & \multirow{2}{*}{17.53} \\ \cline{2-3}
										& [1.0899, 1.0920]\tnote{2} &   11         &                       &                        \\ \hline
										Case II (F-norm)                 & 0.5653 & 100         & {\bf 100 }\%                 & 11.50                  \\ \hline
									\end{tabular}
									\begin{tablenotes}
										\footnotesize
										\item[1] Including the iterations in the first stage and the second stage of Algorithm \ref{alg2-weighting}
										\item[2] Meaning that the $r_{stb}$ falls into this interval.
									\end{tablenotes}
								\end{threeparttable}
							\end{table*}
							\begin{table*}
								\centering
								\caption{Typical successful cases returned via Algorithm \ref{alg2-RSSR}}\label{stability-case-continue}
								\label{typical-cases}
								\begin{tabular}{c|c|c|c}
									\hline
									Cases      & Initial condition      & Typical successful $\Delta$ & Eigenvalues of $A+A_{\Delta}$ \\
									\hline
									{\footnotesize{Case I (2-norm)}}  &  {{\tiny$\tilde \lambda^{(0)}= 1.6763, \tilde \theta^{(0)}=\left[ \begin{matrix}
												0.6035 \\
												0.8478 \\
												0.5046\\
												0.0079
											\end{matrix} \right]$}}  &   {\footnotesize{$\left[ \begin{matrix}
												-0.4973 &  0.1269 \\
												0.1269 &   0.4973	
											\end{matrix} \right]_{r^S_{stb}=0.5132}$}}            &     {\footnotesize {$\{ -1.7523 \pm {\bf j}9.2898, \pm {\bf j}1.3744\}$}    }                        \\
									\hline
									{\footnotesize{ Case I (F-norm)}}  &  {\tiny $\tilde \lambda^{(0)}=-1.8452, \tilde \theta^{(0)}=\left[ \begin{matrix}
											0.2800 \\
											0.3669	\\
											0.3335 \\
											0.7948
										\end{matrix} \right]$} &   {\footnotesize  {$\left[ \begin{matrix}
												-0.0326 &  -0.0707 \\
												0.1978 &   0.4701	
											\end{matrix} \right]_{r^F_{stb}=0.5159}$}}        &       {\footnotesize   {$\{ -1.7922  \pm {\bf j}9.1860, \pm {\bf j}1.3744\}$}    }                        \\
									\hline
									{\footnotesize{Case II (2-norm)}} & {\tiny  $\tilde \lambda^{(0)}=-1.5612, \tilde \theta^{(0)}=\left[ \begin{matrix}
											0.5225 \\
											0 \\
											0 \\
											0.1058
										\end{matrix} \right]$} &  {\footnotesize {$\left[ \begin{matrix}
												-0.5284 &  0 \\
												0 &   0.5284	
											\end{matrix} \right]_{r^S_{stb}=0.5284}$}  }        &    {\footnotesize {$\{-1.7963 \pm {\bf j}9.1147, \pm {\bf j}1.4698\}$}}    \\ \hline
									{\footnotesize{Case II (F-norm)}} & {\tiny  $\tilde \lambda^{(0)}=1.4303, \tilde \theta^{(0)}=\left[ \begin{matrix}
											0.3837 \\
											0 \\
											0 \\
											0.0306
										\end{matrix} \right]$} &  {\footnotesize {$\left[ \begin{matrix}
												-0.0422 &  0 \\
												0 &   0.5638	
											\end{matrix} \right]_{r^F_{stb}=0.5653}$}  }        &    {\footnotesize {$\{-1.7610 \pm {\bf j}8.9371, \pm {\bf j}1.3366\}$}}    \\ \hline
								\end{tabular}
							\end{table*}
							
							\begin{figure}
								\centering
								\includegraphics[width=0.57\linewidth]{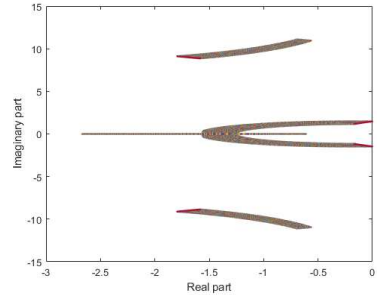}
								\caption{This figure plots the regions of all eigenvalues of $A+E\Delta H$ with $\Delta={\bf diag}\{\theta_1,\theta_2\}$ on $-0.5284\le \theta_1,\theta_2\le 0.5284$. Those regions have two boundary points exactly in the imaginary axis.}
								\label{optimality-verify}
							\end{figure}
							
							\begin{figure}
								\centering
								\includegraphics[width=0.65\linewidth]{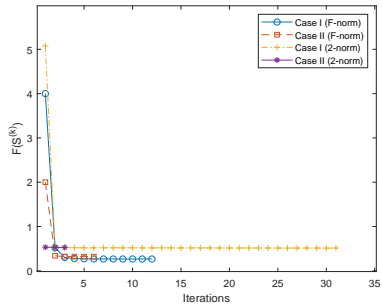}
								\caption{Evolution of $F(Z^{(k)})$. Only the $F(Z^{(k)})$ of the second stage of Algorithm \ref{alg2-RSSR} is presented, since the first stage usually terminates within $2$ or $3$ iterations. }
								\label{iteration-function}
							\end{figure}
							
						\end{example}
						
						\begin{example}
							Consider the following system that appeared in \cite{johnson2018structured} with {\small{ $A=\left[\begin{matrix}
										0&-1& 0& 0 \\
										-1& 0& 1& 0 \\
										1& 0&-1& 0 \\
										-1& 1& 0& 1
									\end{matrix}\right]$, $B=\left[\begin{matrix}
										1 \\
										0 \\
										0 \\
										1
									\end{matrix}\right]$, $A_{\Delta}=\left[\begin{matrix}
										0& 0 & 0& \theta_{14} \\
										0& 0& 0& \theta_{24} \\
										0& 0&\theta_{33}& 0 \\
										\theta_{41}& \theta_{42}& \theta_{43}& \theta_{44}
									\end{matrix}\right]$}}, and  {\small $B_{\Delta}=\left[\begin{matrix}
									\theta_{15} \\
									\theta_{25} \\
									\theta_{35} \\
									\theta_{45}
								\end{matrix}\right]$}.
							We conduct $200$ trials using Algorithm \ref{alg2-weighting} for the RSCR. Let $\gamma$, $\epsilon$, and $\xi$ be selected in the same way as in Example \ref{stability-example}.  In each trial,  we select the initial $\tilde \lambda^{(0)}$ and $\tilde \mu^{(0)}$ randomly from the uniform distribution $(-2,2)$, and $\tilde \theta^{(0)}=[\theta_{41},\theta_{42},\cdots,\theta_{45}]^\intercal$ with its elements randomly from the uniform distribution $(0, 1)$.  We collect the corresponding $r^F_{con}=||\theta||_F$ in Table \ref{controllability-radius-table}. For each returned $\theta$, it has been checked that $\min \nolimits_{z\in {\mathbb{C}}}\sigma_n[A+A_{\Delta}-zI, B+B_\Delta]<6.6738\times 10^{-9}$. This again demonstrates the robustness of the proposed algorithm in converging to the local minima. From Table \ref{controllability-radius-table}, it is found that the obtained $r^F_{con}$ almost falls into three intervals:  \\~ $[0.0027693, 0.0050129]$, $[0.32018, 0.32020]$, and \\~  $[0.64120, 0.64122]$, among which the first interval is of the highest probability ($67.5\%$). It is noted that the third interval is consistent with the result reported in Example 3 of \cite{johnson2018structured}.
							\begin{table*}[!htbp]
								\centering
								\caption{RSCR obtained via Algorithm \ref{alg2-weighting} in $200$ experiments.} \label{controllability-radius-table}
								\begin{tabular}{c|c|c|c}
									\hline
									Interval of $r^F_{con}$ & $[0.0027693,0.0050129]$ & $[0.32018,0.32020]$ & $[0.64120,0.64122]$ \\
									\hline
									Frequencies & 135  & 23  & 43 \\
									\hline
								\end{tabular}
							\end{table*}

							
							Particularly, corresponding to $r^F_{con}=0.0027693$, the associated perturbation is
							{\small
								\begin{equation} \label{perturbation-con1}
									[A_\Delta, B_{\Delta}]\!=\!10^{-3}\left[\begin{matrix}
										0&0&0&0  & 0.0358 \\
										0&0&0&0 &  -0.3825 \\
										0&0  &  -0.6774 &0  &  -0.2945\\
										1.7613  & 0.5038   &  1.6916 &   0.6628  & 0.5647
									\end{matrix}
									\right]	\end{equation}}It can be verified that at $z=-0.0006770\in {\bf \Lambda}(A+A_\Delta)$, $\sigma_{n}([A+A_\Delta-z I, B+B_\Delta])=2.2269\times 10^{-9}$. This RSCR value is less than the one obtained in Example 3 of \cite{johnson2018structured}.
							
							For the RSSZR, with the randomly generated initial condition $\tilde \lambda^{(0)}=1.1209$, $\tilde \mu^{(0)}=-1.6755$, and $\tilde \theta^{(0)}=[0.4509,$\\~
							${\tiny{0.5470,0.2963,0.7447,0.1890,0.6868,0.1835,0.3685,0.6256]^\intercal}}$, we obtain  $r^F_{stz}=0.003178$, with the corresponding 
							{\small
							$$		[A_\Delta, B_{\Delta}]=10^{-3}\left[\begin{matrix}
										0&0&0&0  &  -0.6722 \\
										0&0&0&0 &   -0.1786 \\
										0&0  &  0.0000&0  &  0.1786\\
										1.0970 &  -1.1250  &  1.3720 &   2.2880  & -0.0060
									\end{matrix}
									\right]
							$$} It turns out $A+A_{\Delta}$ has an eigenvalue $z=3.0982\times 10^{-8}$, at which $\sigma_{n}([A+A_\Delta-z I, B+B_\Delta])=6.4983\times 10^{-11}$.  As expected, the $r^F_{stz}$ is larger than the $r^F_{con}$ obtained above. 
							
							
							\begin{figure}[!htbp]
								\centering
								\includegraphics[width=0.68\linewidth]{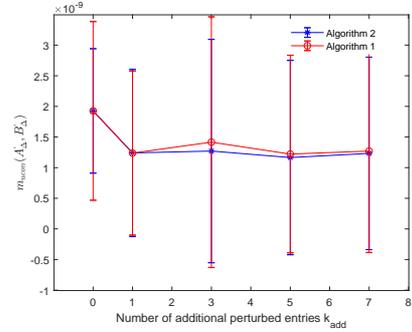}
								\caption{Mean values and standard deviations of $m_{ucon}(A'_{\Delta},B'_\Delta)$}
								\label{compare_quality}
							\end{figure}
							
						\end{example}
						
						
						In order to compare the performance of Algorithm \ref{alg2-weighting} and Algorithm \ref{alg2} (i.e. the algorithm with randomly generated initial conditions), we further allow $k_{add}$ additional entries of $(A,B)$ to be perturbed apart from those in $[A_{\Delta}, B_{\Delta}]$. For each $k_{add}$,  we carry out $200$ experiments by using Algorithm \ref{alg2-weighting} and Algorithm \ref{alg2} to solve the corresponding RSCR problems, in which the positions of the $k_{add}$ additional perturbed entries are randomly selected, and the rest of the algorithmic parameters, as well as the initial points, are set in the same way described above (in each
						experiment, we conduct Algorithms \ref{alg2-weighting} and \ref{alg2} for randomly-generated, identical initial points). For each obtained perturbation $(A'_{\Delta}, B'_{\Delta})$, we adopt $m_{ucon}(A'_{\Delta},B'_\Delta)\doteq \min\nolimits_{z\in {\bf \Lambda}(A+A'_\Delta)}\sigma_n([A+A'_\Delta-z I, B+B'_\Delta])$ as the metric to measure the quality of the perturbation against controllability. Fig. \ref{compare_quality} plots the mean value and its standard deviation of $m_{ucon}(A'_{\Delta}, B'_{\Delta})$ for each $k_{add}$. From this figure, we observe that for the same $k_{add}$, the mean value of $m_{ucon}(A'_{\Delta}, B'_{\Delta})$ via Algorithm \ref{alg2-weighting} is overall equal to or smaller than that of Algorithm \ref{alg2}, while the standard deviation is overall relatively smaller. This indicates the first stage of Algorithm \ref{alg2-weighting} might be helpful in producing more `stable' solutions with better `quality' than Algorithm \ref{alg2}.

						\section{Conclusions}
In this paper, we have addressed the RSCR, RSSZR, and RSSR problems when perturbations depend on the perturbed parameters affinely. Our main contributions are twofold. Firstly, we have proven that checking the feasibility of the RSCR and RSSZR is NP-hard, which implies that global optima of these problems cannot be found in polynomial time unless $P=NP$. Secondly, we have proposed unified rank-relaxation based algorithms for these three problems that exploit the low-rank structure of the problems and use DC programs and SDP solvers. These algorithms are valid for both the $2$-norm and F-norm based problems and are capable of finding local minima with performance specifications on the perturbations under suitable conditions. Numerical simulations have confirmed the effectiveness of our proposed algorithms.

It is worth noting that our methods can be easily extended to handle complex-valued perturbations by changing the real field ${\mathbb R}$ to the complex field ${\mathbb C}$. However, the choice of initial points can influence the obtained local minima, and the computational complexity of our proposed methods is primarily dominated by the adopted convex program solvers. Therefore, future research could explore theoretical approaches for selecting better initial points and developing specialized algorithms that can enhance the computational efficiency of the proposed methods further.

						\section*{Acknowledgements}
						This work was supported in part by the National Natural Science Foundation of China under Grant 62003042.
						
						\section*{References}
						{\footnotesize
							\bibliographystyle{unsrt}
							\bibliography{yuanz3}
						}

					\end{document}